\documentclass[11pt, a4paper, twoside]{article}

\usepackage[utf8]{inputenc}
\usepackage[british]{babel}
\usepackage{appendix}
\usepackage{geometry}
\usepackage{etoolbox}
\usepackage{caption}
\usepackage{makecell}
\usepackage{booktabs}
\usepackage{subcaption}
\usepackage{tabularx}

\usepackage{amsmath}
\usepackage{amsfonts}
\usepackage{amsmath}
\usepackage{amssymb}
\usepackage{mathtools}
\usepackage{amsthm} 
\usepackage{thmtools}
\usepackage{graphicx}
\usepackage{pdfpages}

\usepackage{paralist}
\usepackage{enumitem}

\usepackage[hidelinks]{hyperref}
\usepackage[noabbrev, nameinlink, capitalise]{cleveref}
\hypersetup{ colorlinks=false, linkcolor=black, citecolor=black, filecolor=black, urlcolor=black }

\usepackage{titlesec}
\usepackage{fancyhdr}

\usepackage{natbib}
\setcitestyle{authoryear,open={(},close={)}}

\declaretheorem[style=definition, numberwithin=section]{remark}
\declaretheorem[style=definition, numberwithin=section]{definition}

\newtheorem{proposition}{Proposition}[section]

\declaretheorem[style=definition ,numberwithin=section]{example}
\newtheorem{theorem}{Theorem}[section]

\newcommand{\R}{\mathbb{R}}

\newcommand{\Om}{\Omega}
\newcommand{\Sig}{\mathcal{A}}
\newcommand{\Prob}{\mathbb{P}}
\newcommand{\E}{\mathbb{E}}

\usepackage{scalerel}[2016/12/29]

\titleformat{\section}{\normalfont\scshape\large\centering}{\thesection}{1em}{}

\usepackage{fancyhdr}
\fancypagestyle{plain}{%
	\fancyhf{}
	\fancyfoot[RO,LE]{\rm\thepage}

}

\setlist[itemize]{noitemsep}
\setlist[enumerate]{noitemsep}

\title{ESG Risk: Lessons Learned from Utility Theory}
\date{\today}
\author{%
    Sebastian Geissel%
    \thanks{%
        Trier University of Applied Sciences,
        Business School,
        Schneidershof,
        D--54293 Trier, Germany.
        E-Mail: \url{s.geissel@hochschule-trier.de}
    }
    \and
    Christoph Knochenhauer%
    \thanks{%
        Technical University of Munich,
        School of Computation, Information and Technology,
        Parkring 11,
        D--85748 Munich, Germany.
        E-Mail: \url{knochenhauer@tum.de}
}
}

\newcommand{\ESGSpace}{\mathbb{S}}
\newcommand{\RiskSF}{\hat{\rho}}
\newcommand{\RiskESG}{\rho}
\newcommand{\loss}{\ell}
\newcommand{\util}{u}
\newcommand{\positions}{\mathcal{Q}}

\newcommand{\xlow}{\underline{x}}
\newcommand{\slow}{\underline{s}}
\newcommand{\indifference}{\mathcal{I}}
\newcommand{\premium}{\mathfrak{p}}
\newcommand{\defined}{\coloneqq}

\newcommand{\argdot}{\,\cdot\,}

\newcommand{\bigSetSep}{\;\big|\;}
\newcommand{\BigSetSep}{\;\Big|\;}

\newcommand{\effDom}{\mathrm{dom}}
\newcommand{\raw}{\mathrm{raw}}
\newcommand{\norma}{\mathrm{norm}}

\newcommand{\close}{\hspace*{\fill}\emph{$\boldsymbol\diamond$}}
\newcommand{\closeEqn}{\tag*{\emph{$\boldsymbol\diamond$}}}

\newtheorem*{standingassumption}{Standing Assumption}

\begin{document}

\maketitle

\section*{Abstract} \label{sec:abstract}
We propose a new class of monetary risk measures for assessing financial and ESG risk.
The construction is based on classical shortfall risk measures with loss function replaced by a multi-attribute utility function.
We present an extensive theoretical analysis of these risk measures, showing specifically how properties of the utility function translate into properties of the associated risk measure.
We furthermore discuss how these multi-attribute risk measures can be used to compute minimum risk portfolios and show in a numerical study that accounting for ESG risk in optimal portfolio choice has a significant influence on the composition of portfolios.

\section{Introduction} \label{sec:introduction}

At the core of the base operations in the financial industry lies the management and mitigation of risks.
The growing importance of environmental, social, and governance (ESG) considerations therefore raises the question of how ESG-related risks can be incorporated into established risk management frameworks, including from a regulatory perspective; see \cite{ecb2020, boe2023, fed2023, eiopa2024, hkma2024, eba2025}.

Beyond anecdotal evidence from industry practice,%
\footnote{In his 2022 \emph{Letter To CEOs}, the BlackRock CEO Larry Fink calls sustainable investing ``a tectonic shift of capital'', noting that ``sustainable investments have now reached \$4 trillion'', referring to data sources from Morningstar, Simfund, Broadridge; see \url{https://www.blackrock.com/corporate/investor-relations/2022-larry-fink-ceo-letter}.}
\footnote{Since the launch of the \emph{United Nations Principles for Responsible Investment} (PRI) in 2005, the number of signatories has grown from 50 in 2005 to more than 5,000 in 2024, ``representing over half the world’s institutional assets under management''; see \url{https://www.unpri.org/pri-blog/meeting-investors-where-theyre-at-pri-launches-new-strategy/12636.article}.}
\footnote{As of 2025--01--01, the global sustainable fund universe accounted for \$3.2 trillion in assets, representing an 8\% increase over the previous year and a tripling over the last five years; see \url{https://www.morningstar.com/business/insights/research/global-esg-flows}.}
\emph{ESG-Finance} has become a vibrant area of academic research.
One important observation is that explicitly accounting for ESG reveals a more nuanced picture on financial decision making.
While empirical studies document that sustainability aspects have an effect on investment decisions (\cite{hartzmark2019investors}), the implications for portfolio performance remain contested.
Parts of the literature argue that ESG has a detrimental effect on financial portfolios (e.g., \cite{hong2009price, bolton2021investors, cornell2021esg, hsu2023pollution}), whereas others find neutral (\cite{naffa2022, pastor2022dissecting}) or even positive effects (\cite{kotsantonis2016esg, lins2017social}).
An explanation of these diverging findings is that, as \cite{pedersen2021responsible} put it, ``... ESG is not fully priced in the market''.
Similarly, a key implication of the findings of \cite{avramov2022} is that purely financial measures may fail to capture an investment’s full risk profile.

This discussion allows for but one conclusion: risk assessment should explicitly account for ESG considerations.
This is the central theme of the present paper.
We propose a framework for measuring combined financial and ESG risk, inspired by the axiomatic approach of \cite{artzner1999},
and extending the classical theory of monetary risk measures developed by \cite{follmer2002}.
In the classical setting, a large class of convex monetary risk measures can be expressed as
\[
  \RiskSF[X] \defined \inf\bigl\{m\in\R:
    \E\bigl[\loss(-X-m)\bigr] \leq x_0\bigr\},
\]
where $X$ is the net financial value of the position to be evaluated, $\ell:\R\to\R$ an increasing and convex loss function, and $x_0\in\R$ a fixed reference level.
Our observation is that any such risk measure can equivalently be expressed as
\[
  \RiskSF[X]\defined\inf\bigl\{m\in\R:\E\bigl[\hat\util(X+m)\bigr]\geq 0\bigr\},
\]
for a concave increasing utility function $\hat\util(x) \defined -\ell(-x)+x_0$.
This allows us to interpret $\RiskSF[X]$ as the minimum amount of cash $m\in\R$ such that $X+m$ has positive expected utility.
We generalize this framework by allowing investors to account for an additional non-financial risk factor $S$, accounting e.g.\ for ESG factors, through a multi-attribute utility function
$(x,s)\mapsto\util(x,s)$. This leads to an ESG risk measure of the form
\[
  \RiskESG[X,S]
  \defined \inf\bigl\{m\in\R:\E\bigl[\util(X+m,S)\bigr]\geq 0\bigr\}.
\]
An advantage of our approach is that it preserves the axiomatic foundations of utility theory,
allowing assumptions on the interaction between financial and ESG risk to be directly encoded in the structure of the utility function $\util$ and, consequently, in the risk measure $\RiskESG[X,S]$.
This contrasts with the existing literature, where ESG risk is often incorporated through ad hoc extensions of classical risk measures.

Related approaches include using ESG score dispersion to explain financial volatility (\cite{capelli2021forecasting}),
extending value at risk with ESG information (\cite{capelli2023}),
or applying coherent risk measures to convex combinations of financial returns and ESG scores (\cite{torri2023esg}).
Alternative frameworks include multi-risk-aversion utility models (\cite{escobar2022}),
vector- or set-valued risk measures (\cite{jouini2004vector, hamel2010duality}). A systematic review of the literature on the conceptualization of ESG risks in portfolio studies is provided by \cite{gallucci2022}. 

Finally, from a utility-based perspective, \cite{bollen2007}, \cite{jessen2012}, \cite{dorfleitner2017} are closely related to our approach in that they employ additive multi-attribute utility functions of the form $\util(x,s) = \util_1(x) + \util_2(s)$ in optimal investment problems. For risk measurement, however, this additive specification implies that the joint distribution of financial and ESG components is effectively ignored. In contrast, we assume only mutual utility independence, leading to the specification $\util(x,s) = \util_1(x) + \util_2(s) + k\util_1(x)\util_2(s)$ for some $k\in\R$, which explicitly accounts for dependence between $X$ and $S$ when measuring risk.

The main contributions of this paper are threefold:
\begin{enumerate}
\item We develop a rigorous mathematical framework for measuring risk based on multiple attributes and apply this framework to construct a class of risk measures combining financial and ESG risk.
\item Financial and ESG exposure are modeled as separate, possibly dependent, random variables and combined via a multi-attribute utility function.
This allows for a clear separation between the modeling of risk exposures and investor preferences.
\item Using ESG ratings and stock returns for all S\&P 500 firms over an extended sample period, we provide empirical evidence that the proposed risk measures are capable of painting a fine-grained picture of the inherent risk profile of assets with similar financial performance but differing ESG exposures.
At the portfolio level, minimum-risk strategies based on ESG risk measures achieve substantially improved ESG profiles without sacrificing returns.
\end{enumerate}

Section~\ref{sec:risk_measures_for_ESG_assets} develops the utility-based risk measures and establishes their key properties, such as translation invariance, monotonicity, and convexity, along with several examples of ESG risk measures.
It also introduces ESG risk premia and characterizes favorable and unfavorable ESG exposures.
Section~\ref{sec:applications_of_esg_risk_measures} applies the framework to a dataset comprising all S\&P 500 companies, examining single-asset risk measurement and minimum-risk portfolio construction.

\section{Risk Measures for ESG Assets}\label{sec:risk_measures_for_ESG_assets}

We propose a class of risk measures that accounts for both financial and ESG risk of an investment position.
Financial exposure is represented by a random variable $X:\Om\to\R$ on a probability space $(\Om,\Sig,\Prob)$, where $X(\omega)$ denotes the net value (in monetary units) in scenario $\omega\in\Om$.
ESG exposure is modeled by a random variable $S:\Om\to\ESGSpace$, with $\ESGSpace\subseteq\R$, where $S(\omega)$ is the future ESG rating in scenario $\omega$.
Typically, $\ESGSpace$ is either a finite set, e.g.\ $\{1,\dots,N\}$ for rating classes, or a bounded interval such as $[0,1]$ or $[0,100]$ for score-based ratings.

\subsection{Utility-Based Shortfall Risk Measures}

To motivate our definition of ESG risk measures, we recall the shortfall risk measures introduced by \cite{follmer2002}.
Given an increasing and convex loss function $\loss:\R\to\R$ and a threshold $x_0\in\R$, the shortfall risk of a financial position $X$ is defined as
\[
  \RiskSF[X] \defined \inf\bigl\{m\in\R: \E\bigl[\loss(-X-m)\bigr] \leq x_0\bigr\}.
\]
A position is acceptable if $\E[\loss(-X)]\leq x_0$; otherwise, cash $m$ is added to make $X+m$ acceptable.
Hence, $\RiskSF[X]$ represents the minimal capital requirement ensuring acceptability.

Defining $\hat\util\defined-\loss(-\argdot)+x_0$, the shortfall risk measure can be written as
\[
  \RiskSF[X]\defined\inf\bigl\{m\in\R:\E\bigl[\hat\util(X+m)\bigr]\geq 0\bigr\}.
\]
The function $\hat\util:\R\to\R$ is increasing and concave, and thus a utility function.
Acceptable positions therefore correspond to non-negative expected utility.
This suggests a natural generalization to multiple risk dimensions by replacing $\hat\util$ with a multi-attribute utility function $\util$, leading to the following definition of ESG shortfall risk measures.

\begin{definition}\label{def:esg_risk_measures}
For a measurable function $\util:\R\times\ESGSpace\to\R\cup\{-\infty\}$, define%
\footnote{Here, $\util_+\defined \max\{\util,0\}$ denotes the positive part of $\util$. The integrability condition ensures that $\E[\util(X+m,S)]$ exists, but may equal $-\infty$.}
the set of \emph{quantifiable positions} by
\[
  \positions \defined \Bigl\{(X,S):\Om\to\R\times\ESGSpace\text{ measurable} \BigSetSep \E\bigl[\util_+(X+m,S)\bigr] < \infty \text{ for all }m\in\R\Bigr\}.
\]
The \emph{shortfall risk measure for ESG assets} is then given by
\[
  \RiskESG[X,S] \defined \inf\bigl\{m\in\R:\E\bigl[\util(X+m,S)\bigr]\geq 0\bigr\},
  \qquad (X,S)\in\positions.\closeEqn
\]
\end{definition}

As in the single-attribute case, a position $(X,S)$ is acceptable if $\E[\util(X,S)]\geq 0$.
Otherwise, adding cash to $X$ renders it acceptable, and $\RiskESG[X,S]$ is the minimal required amount.

Just as for classical shortfall risk measures, $\RiskESG[X,S]$ is translation invariant in $X$, allowing us to interpret $X\mapsto\RiskESG[X,S]$ as a monetary risk measure.

\begin{theorem}[Translation Invariance]\label{theorem:translation_invariance}
Any ESG risk measure $\RiskESG$ satisfies
\[
  \RiskESG[X + \eta,S] = \RiskESG[X,S] - \eta\qquad\text{for all }\eta\in\R
\]
and any $(X,S)\in\positions$.
In particular, if $\RiskESG[X,S]\in\R$, then $\RiskESG[X + \RiskESG[X,S],S] = 0$.\close
\end{theorem}

\begin{proof}
By definition,
\begin{align*}
  \RiskESG[X + \eta,S]
  &= \inf\bigl\{m\in\R : \E[\util(X+\eta+m,S)]\geq 0\bigr\}\\
  &= \inf\bigl\{\tilde m\in\R : \E[\util(X+\tilde m,S)]\geq 0\bigr\}-\eta
   = \RiskESG[X,S]-\eta,
\end{align*}
which proves the first claim. The second follows immediately by substituting $\eta=\RiskESG[X,S]$.
\end{proof}

Observe that the utility function $\util$ in the definition of $\RiskESG$ may take the value $-\infty$, which allows to declare positions with ESG ratings below a threshold $\bar s\in\ESGSpace$ unacceptable.
This is the case whenever
\[
  \util(x,s) = -\infty
  \qquad\text{for all }(x,s)\in\R\times\ESGSpace\text{ with } s<\bar s.
\]
This leads to the following example.

\begin{example}[Penalty and Threshold ESG Risk Measures]
Let $\util_1:\R\to\R\cup\{-\infty\}$ be measurable and denote by $\RiskSF$ the associated financial shortfall risk measure
\[
  \RiskSF[X]\defined\inf\bigl\{m\in\R:\E[\util_1(X+m)]\geq 0\bigr\}.
\]
For $\bar s\in\ESGSpace$ and $\eta\in[0,+\infty]$, define $\util_2:\ESGSpace\to\R\cup\{-\infty\}$ by
\[
  \util_2(s)\defined
  \begin{cases}
    0, & s\geq \bar s,\\
    -\eta, & s<\bar s.
  \end{cases}
\]
Assume that the multi-attribute utility is additive, i.e.\ $\util(x,s)=\util_1(x)+\util_2(s)$.
The resulting ESG risk measure then satisfies
\[
  \RiskESG[X,S]=\RiskSF[X]+\eta\Prob[S<\bar s],
  \qquad (X,S)\in\positions.
\]
Thus, $\RiskESG$ coincides with $\RiskSF$ for positions with $S\geq \bar s$, while positions facing a drop below $\bar s$ incur an additional penalty $\eta\Prob[S<\bar s]$. Choosing $\eta=+\infty$ renders such positions unacceptable.
We therefore refer to $\RiskESG$ as a \emph{penalty ESG risk measure} for $\eta<+\infty$ and as a \emph{threshold ESG risk measure} for $\eta=+\infty$.\close
\end{example}

An advantage of the utility-based approach to shortfall risk measures is that assumptions on the interaction of financial and ESG risk can be translated into structural properties of the utility function $\util$ via multi-attribute utility theory.
Throughout, we impose the following assumption.

\begin{standingassumption}
The attributes ``financial risk'' and ``ESG risk'' are mutually utility
independent.\close
\end{standingassumption}

To recall the notion of utility independence, let $X_1,X_2$ be two financial positions.
Financial risk is utility independent of ESG risk if the preference for $X_2$ over $X_1$ does not depend on the fixed ESG level, that is,
\[
  \E[u(X_1,\bar s)] \le \E[u(X_2,\bar s)]
  \quad\text{for some }\bar s\in\ESGSpace
\]
if and only if
\[
  \E[u(X_1,s)] \le \E[u(X_2,s)]
  \quad\text{for all } s\in\ESGSpace,
\]
whenever the corresponding positions are quantifiable.
Mutual utility independence means that both attributes are utility independent of each other.
A classical result in multi-attribute utility theory (see \cite{keeney1993decisions}) states that this holds if and only if there exist $k\in\R$ and single-attribute utility functions
\[
  \util_1:\R\to\R\cup\{-\infty,+\infty\},\qquad
  \util_2:\ESGSpace\to\R\cup\{-\infty,+\infty\},
\]
such that%
\footnote{We adopt the convention $\infty-\infty=-\infty$.}
\[
  \util(x,s)=\util_1(x)+\util_2(s)+k\util_1(x)\util_2(s),
  \qquad (x,s)\in\R\times\ESGSpace.
\]
Hence, we subsequently assume that $\util$ admits this representation.

\begin{remark}\label{rem:additive}
For $k\neq 0$, the representation of $\util$ is equivalent to
\[
  1 + k\util(x,s) = \bigl(1 + k\util_1(x)\bigr)\bigl(1 + k\util_2(s)\bigr),
  \qquad (x,s)\in\R\times\ESGSpace,
\]
that is, a linear transformation of $\util$ is multiplicative in the corresponding transformations of $\util_1$ and $\util_2$.
If $k=0$, the representation is additive,
\[
  \util(x,s)=\util_1(x)+\util_2(s).
\]
This additive form is commonly assumed in the literature; see \cite{bollen2007, jessen2012, dorfleitner2017}.
As shown in \cite{keeney1993decisions}, it is equivalent to \emph{additive independence} of financial and ESG risk, meaning that preferences depend only on the marginal distributions of $X$ and $S$ and not on their joint distribution.
Consequently, the risk of $(X,S)$ is evaluated as if $X$ and $S$ were independent.
While additive utilities are useful, restricting to this class alone is a strong assumption.\close
\end{remark}

\subsection{Monotonicity of ESG Risk Measures}

A fundamental property of risk measures is monotonicity: higher financial payoffs or improved ESG ratings should reduce risk.
The utility-based framework allows to translate this requirement into simple conditions on $\util$.

\begin{theorem}[Monotonicity]\label{theorem:monotonicity_risk}
Assume that $\util$ is non-decreasing in both arguments and let $(X,S),(X',S')\in\positions$ satisfy $X\leq X'$ and $S\leq S'$.
Then
\[
  \RiskESG[X,S] \geq \RiskESG[X',S'].\closeEqn
\]
\end{theorem}

\begin{proof}
Since $\util$ is non-decreasing,
\[
  \E[\util(X+m,S)] \geq \E[\util(X'+m,S')] \quad\text{for all } m\in\R,
\]
which directly implies the claim.
\end{proof}

It is natural to construct the multi-attribute utility function $\util$ from the single-attribute utilities $\util_1$ and $\util_2$.
Requiring $\util$ to be monotone therefore imposes conditions on $\util_1$ and $\util_2$, which are themselves assumed to be monotone as they model preferences toward financial and ESG exposures.
The following result characterizes this relationship.

\begin{proposition}[Monotonicity of $\util$]\label{prop:monotone_utility}
Suppose that $\util_1$ is non-decreasing. Then, for any $s\in\ESGSpace$,
\[
  x \mapsto \util(x,s)\text{ is non-decreasing on $\R$}
  \qquad\text{if and only if}\qquad
  1 + k\util_2(s)\in\{-\infty\}\cup[0,\infty].
\]
Analogously, if $\util_2$ is non-decreasing and $x\in\R$ is fixed,
\[
  s \mapsto \util(x,s)\text{ is non-decreasing on $\ESGSpace$}
  \qquad\text{if and only if}\qquad
  1 + k\util_1(x)\in\{-\infty\}\cup[0,\infty].\closeEqn
\]
\end{proposition}

\begin{proof}
Using the representations
\[
  \util(x,s)=\bigl(1+k\util_2(s)\bigr)\util_1(x)+\util_2(s),\qquad
  \util(x,s)=\bigl(1+k\util_1(x)\bigr)\util_2(s)+\util_1(x),
\]
the claim follows immediately.
\end{proof}

\begin{remark}
Proposition~\ref{prop:monotone_utility} shows that monotonicity of $\util,\util_1$, and $\util_2$ restricts the ranges of
\[
  x \mapsto 1+k\util_1(x)
  \qquad\text{and}\qquad
  s \mapsto 1+k\util_2(s).
\]
If $k>0$, this is equivalent to
\[
  \util_1,\util_2 \text{ taking values in } \{-\infty\}\cup[-1/k,\infty).
\]
While a lower bound on $\util_2$ is typically unproblematic since $\ESGSpace$ is bounded, imposing such a bound on $\util_1$ is restrictive, as it caps the risk associated with financial tail events.
In this case, it may be preferable to relax monotonicity in $s$ for small $x$.

If $k<0$, monotonicity instead requires
\[
  \util_1,\util_2 \text{ taking values in }(-\infty,-1/k]\cup\{+\infty\}.
\]
For $k=0$, no range restrictions arise since $\util(x,s)=\util_1(x)+\util_2(s)$.\close
\end{remark}

In view of the discussion above, monotone multi-attribute utilities can be constructed from monotone single-attribute utilities as follows.
Fix $k\in\R$ and non-decreasing functions $\hat{\util}_1:\R\to\R\cup\{-\infty\}$ and $\hat{\util}_2:\ESGSpace\to\R\cup\{-\infty\}$, and define the effective domain
\[
  \effDom(\util)=\bigl\{(x,s)\in\R\times\ESGSpace:
  1+k\hat{\util}_1(x)\ge 0,\; 1+k\hat{\util}_2(s)\ge 0\bigr\}.
\]
Then set
\[
  \util(x,s)=
  \begin{cases}
    \hat{\util}_1(x)+\hat{\util}_2(s)+k\hat{\util}_1(x)\hat{\util}_2(s),
      & (x,s)\in\effDom(\util),\\
    -\infty, & \text{otherwise},
  \end{cases}
\]
so that $\util$ is finite precisely on $\effDom(\util)$.
We refer to this procedure as the \emph{canonical construction}.

\begin{example}[Entropic and Capped Entropic ESG Risk Measures]\label{ex:entropic_Esg_Risk_measure}
As an illustration of the canonical construction, consider exponential single-attribute utilities with parameters $\gamma_1,\gamma_2>0$,
\[
  \hat{\util}_1(x)\defined \frac{1}{\gamma_1}\bigl(1-e^{-\gamma_1 x}\bigr),
  \qquad
  \hat{\util}_2(s)\defined c\frac{1}{\gamma_2}\bigl(1-e^{-\gamma_2 (s-s_0)}\bigr),
\]
where $c>0$ is a scaling factor and $s_0\in\ESGSpace$ a baseline level.
Note that $\hat{\util}_2(s)\ge0$ if and only if $s\ge s_0$, so that $s_0$ normalizes ESG
utility.

The financial shortfall risk measure associated with $\hat{\util}_1$ is the entropic risk measure
\begin{equation}\label{eq:entropic}
  \RiskSF[X]=\frac{1}{\gamma_1}\log\E\bigl[e^{-\gamma_1 X}\bigr];
\end{equation}
see \cite{follmer2002}.
To construct a multi-attribute utility, we first drop global monotonicity and set
\begin{equation}\label{eq:utility_uncapped_entropic}
  \util(x,s)=\hat{\util}_1(x)+\hat{\util}_2(s)+k\hat{\util}_1(x)\hat{\util}_2(s),
\end{equation}
which yields the \emph{entropic ESG risk measure}.
For $k>0$, monotonicity in $x$ holds whenever $s\ge\slow$, where
\[
  \slow\defined -\frac{1}{\gamma_2}\log\Bigl(1+\frac{\gamma_2}{ck}\Bigr)+s_0,
\]
and monotonicity in $s$ holds whenever $x\ge\xlow$, with
\[
  \xlow\defined -\frac{1}{\gamma_1}\log\Bigl(1+\frac{\gamma_1}{k}\Bigr).
\]

Alternatively, applying the canonical construction yields a globally monotone utility $\util_{cap}$ with effective domain
\[
  \effDom(\util_{cap})=\bigl([\xlow,\infty)\times[\slow,\infty)\bigr)\cap
  (\R\times\ESGSpace),
\]
and equal to $-\infty$ outside this set.
The associated risk measure is called the \emph{capped entropic ESG risk measure}.
Although this may assign infinite risk to some positions, boundedness of financial positions $X$, which is commonly assumed in the literature, implies that $\RiskESG$ is infinite only if $\Prob[S<\slow]>0$, that is, when ESG exposure renders the position unacceptable.
Thus, the lower bounds on $\util_1$ and $\util_2$ are not overly restrictive.\close
\end{example}

\subsection{ESG Risk Premia and Indifference Positions}

To analyze the impact of ESG exposure $S$ on risk evaluation, we compare the ESG risk measure $\RiskESG$ associated with $\util=\util_1+\util_2+k\util_1\util_2$ to the purely financial shortfall risk measure
\[
  \RiskSF[X]\defined\inf\bigl\{m\in\R:\E[\util_1(X+m)]\ge0\bigr\}.
\]
We define the \emph{ESG risk premium} by
\[
  \premium[X,S]\defined \RiskESG[X,S]-\RiskSF[X], \qquad (X,S)\in\positions.
\]
A key role is played by the set of ESG exposures
\begin{multline*}
  \indifference\defined \Bigl\{S:\Om\to\ESGSpace \BigSetSep \E[\util_2(S)]\in\R \text{ and } \E[\util(X,S)]=\E[\util_1(X)]\\
    \text{for all }X:\Om\to\R \text{ with }(X,S)\in\positions\Bigr\},
\end{multline*}
called the \emph{indifference ESG positions}.
For $S\in\indifference$,
\[
  \RiskESG[X,S]=\RiskSF[X], \qquad \text{and hence } \premium[X,S]=0.
\]
Thus, $\indifference$ consists of ESG exposures for which ESG considerations do not affect risk.
The condition $\E[\util_2(S)]\in\R$ ensures quantifiability of ESG risk for constant financial positions.

\begin{proposition}[Characterization of Indifference ESG Positions]\label{proposition:indifference}
Assume there exists $\bar x\in\R$ such that $1+k\util_1(\bar x)\in(0,\infty)$.
Then
\[
  \indifference=\bigl\{S:\Om\to\ESGSpace \bigSetSep S \text{ measurable and } \E[\util_2(S)]=0\bigr\}.
\]
Moreover, if $\ESGSpace$ is an interval and $\util_2$ is continuous and non-decreasing on $\{\util_2\neq-\infty\}$, the set $\indifference$ is non-empty if and only if there exists $s_0\in\ESGSpace$ with $\util_2(s_0)=0$.\close
\end{proposition}

\begin{proof}
Let $(X,S)\in\positions$ with $\E[\util_2(S)]$ existing and
\[
  \E[\util_1(X)]=\E[\util(X,S)]
  =\E[\util_1(X)+\util_2(S)+k\util_1(X)\util_2(S)].
\]
This is equivalent to
\[
  \E\bigl[\util_2(S)\bigl(1+k\util_1(X)\bigr)\bigr]=0.
\]
Hence $|1+k\util_1(X)|=\infty$ must have zero probability, implying $1+k\util_1(X)\in[0,\infty)$.
In particular, for constant $X=\bar x$ with $1+k\util_1(\bar x)>0$ we obtain
\[
  S\in\indifference \quad\Longleftrightarrow\quad \E[\util_2(S)]=0.
\]
The final claim follows from the mean value theorem.
\end{proof}

Since $\RiskESG$ is non-increasing whenever $\util$ is non-decreasing, the set of indifference ESG positions separates positive and negative ESG risk premia.
Indeed, for $(X,S),(X',S')\in\positions$ with $X\le X'$ and $S\le S'$, Theorem~\ref{theorem:monotonicity_risk} yields $\RiskESG[X',S']\le \RiskESG[X,S]$.
In particular, for $S'\in\indifference$,
\[
  S\ge S' \quad\Rightarrow\quad \RiskESG[X,S]\le \RiskESG[X,S']=\RiskSF[X].
\]
Thus, indifference positions act as thresholds: favorable ESG exposures reduce risk, while unfavorable ones increase it.
In particular, if $s_0\in\ESGSpace$ satisfies $\util_2(s_0)=0$, then $s_0$ separates favorable from unfavorable ESG outcomes.

\begin{example}[Calibration of ESG Risk Measures]\label{ex:calibration}
Indifference ESG positions provide a natural tool for calibrating ESG risk measures.
Consider again the (uncapped) entropic ESG risk measure with
\[
  \util(x,s)=\util_1(x)+\util_2(s)+k\util_1(x)\util_2(s),
\]
where
\[
  \util_1(x)=\frac{1}{\gamma_1}\bigl(1-e^{-\gamma_1 x}\bigr),\qquad
  \util_2(s)=c\frac{1}{\gamma_2}\bigl(1-e^{-\gamma_2(s-s_0)}\bigr).
\]
The parameter $\gamma_1$ determines the classical financial entropic risk measure
\[
  \RiskSF[X]=\frac{1}{\gamma_1}\log\E\bigl[e^{-\gamma_1 X}\bigr],
\]
and is typically fixed in advance.
The baseline $s_0$ is chosen as the unique ESG level satisfying $\util_2(s_0)=0$, which corresponds to an indifference position separating favorable and unfavorable ESG exposures.
Using an additional random indifference position $S\in\indifference$ allows to identify the ESG risk aversion parameter $\gamma_2$ via the condition $\E[\util_2(S)]=0$.

Finally, the parameters $c>0$ and $k\in\R$ control the magnitude of ESG effects.
Rewriting
\begin{equation}\label{eq:calibration}
  \util(x,s)=\bigl(1+k\util_2(s)\bigr)\util_1(x)+\util_2(s),
\end{equation}
shows that $\util_2(s)$ induces both an additive shift and a relative scaling of the financial utility $\util_1(x)$.
The scaling factor $c$ determines the maximal absolute adjustment, while $k$ governs the strength of the relative scaling.\close
\end{example}

Using the same argument as in the proof of Proposition~\ref{proposition:indifference}, we obtain the following result in the case of independent financial and ESG exposures.

\begin{proposition}[Favorable and Unfavorable ESG Exposures]\label{prop:monotone_expectation}
Let $(X,S)\in\positions$ with $X$ and $S$ independent and assume that $\util_2$ and $s\mapsto\util(x,s)$ are non-decreasing for all $x\in\R$.
Then
\[
  \E[\util_2(S)] \le 0 \;\Rightarrow\; \RiskESG[X,S]\ge \RiskSF[X]
  \quad\text{and}\quad
  \E[\util_2(S)] \ge 0 \;\Rightarrow\; \RiskESG[X,S]\le \RiskSF[X].\closeEqn
\]
\end{proposition}

\begin{proof}
We prove only the first implication.
It suffices to show that
\[
  \big\{m:\E[\util(X+m,S)]\geq0\big\}\subseteq \big\{m:\E[\util_1(X+m)]\geq 0\big\}.
\]
Fix $m$ with $\E[\util(X+m,S)]\ge0$.
Then
\begin{align*}
  0 \leq \E[\util(X+m,S)]
  &= \E\bigl[\util_1(X+m)\bigr]
    + \E\bigl[\util_2(S)\bigl(1 + k\util_1(X+m)\bigr)\bigr]\\
  &= \E\bigl[\util_1(X+m)\bigr]
    + \E\bigl[\util_2(S)\bigr]\E\bigl[1 + k\util_1(X+m)\bigr],
\end{align*}
by independence.
Since $\util(X+m,S)\neq-\infty$ a.s., Proposition~\ref{prop:monotone_utility} implies $\E[1+k\util_1(X+m)]\ge0$.
If $\E[\util_2(S)]\le0$, we obtain $0\le\E[\util_1(X+m)]$, which proves the claim.
\end{proof}

The independence assumption can be relaxed by requiring
\[
  \E[\util_2(S)\bigl(1+k\util_1(X+m)\bigr)]\le0
  \quad\text{for all }m\text{ with }\E[\util(X+m,S)]\ge0,
\]
which again yields $\RiskESG[X,S]\ge\RiskSF[X]$.
This condition, however, depends on the joint distribution of $(X,S)$ and is less convenient to verify.

\subsection{Convex ESG Risk Measures}

We conclude by deriving conditions ensuring convexity of $\RiskESG$, which naturally follow from concavity properties of the utility function $\util$.

\begin{theorem}[Convexity]
Assume that $x\mapsto\util(x,s)$ is concave for all $s\in\ESGSpace$ and let $S:\Om\to\ESGSpace$.
Then
\[
  X\mapsto \RiskESG[X,S]\quad\text{is convex on } \{X:\Om\to\R\mid (X,S)\in\positions\}.
\]
If $\util$ is jointly concave, then
\[
  (X,S)\mapsto \RiskESG[X,S]\quad\text{is convex on }\positions.\closeEqn
\]
\end{theorem}

\begin{proof}
Let $\lambda\in[0,1]$ and $(X_1,S_1),(X_2,S_2)\in\positions$, assuming $S_1=S_2$ when concavity holds only in $x$.
If either $\RiskESG[X_i,S_i]=\infty$, the claim is trivial; otherwise choose $m_i\in\R$ such that
\[
  \E[\util(X_i+m_i,S_i)]\ge0,\qquad i=1,2.
\]
By concavity of $\util$,
\begin{align*}
  0 &\le \lambda\E[\util(X_1+m_1,S_1)] + (1-\lambda)\E[\util(X_2+m_2,S_2)]\\
    &\le \E\bigl[\util(\lambda X_1+(1-\lambda)X_2+m_\lambda,\lambda S_1+(1-\lambda)S_2)\bigr],
\end{align*}
with $m_\lambda=\lambda m_1+(1-\lambda)m_2$.
Hence,
\[
  \RiskESG[\lambda X_1+(1-\lambda)X_2,\lambda S_1+(1-\lambda)S_2]
  \le m_\lambda,
\]
and minimizing over $m_1,m_2$ yields convexity.
\end{proof}

While convexity in the financial component $X$ is desirable as it promotes diversification, it is less clear whether a risk measure should also be convex in the ESG component.
Since ESG exposures are often harder to quantify, it is reasonable to allow more general, possibly non-concave utilities $\util_2$, such as S-shaped functions of the form
\[
  \util_2(s) = \begin{cases}
    (s-s_0)^\gamma & \text{if }s\geq s_0\\
    -\lambda(s-s_0)^\gamma & \text{if }s\leq s_0
  \end{cases}\qquad s\in\ESGSpace.
\]
Here, the reference point $s_0\in\ESGSpace$ satisfies $\util_2(s_0) = 0$, that is, it is the benchmark ESG level separating favorable from unfavorable exposures in the spirit of Proposition~\ref{prop:monotone_expectation}.

\section{ESG Risk of S\&P 500 Companies and Minimum Risk Portfolios} \label{sec:applications_of_esg_risk_measures}

In this section, we evaluate ESG risk measures empirically using S\&P 500 stock prices and Sustainalytics ESG Risk Ratings from October 2021 to January 2025. We focus on the uncapped entropic ESG risk measure (Example~\ref{ex:entropic_Esg_Risk_measure}), compare it with the classical entropic risk measure (Equation~\eqref{eq:entropic}), and analyze how uncertainty in ESG ratings affects risk assessments. Finally, we study the impact of ESG risk on portfolio selection by computing and comparing minimum risk portfolios with and without ESG risk.

\subsection{Data Description and Preparation}\label{subsec:data_description}

Our dataset consists of monthly \emph{Sustainalytics ESG Risk Ratings} and adjusted closing stock prices for
S\&P 500 companies from 2021--10--01 to 2025--01--01, obtained from \emph{Bloomberg}. Companies with incomplete data (usually due to not being a constituent of the S\&P 500 over the full time period) are excluded, resulting in 40 monthly observations for a total of 483 companies.

The Sustainalytics ESG Risk Ratings range from 0 (no ESG risk) to 100 and are classified into five categories:
\[
  \text{Negligible: 0--9,}
  \qquad
  \text{Low: 10--19,}
  \qquad
  \text{Medium: 20--29,}
  \qquad
  \text{High: 30--39,}
  \qquad
  \text{Severe: 40+.}
\]
Table~\ref{tab:summary} provides summary statistics of the data set.

\begin{table}[ht]
    \captionsetup{position=above, singlelinecheck=false, labelfont=bf}
    \caption{\textbf{Descriptive statistics of the \emph{Sustainalytics ESG Risk Ratings} data set}}
    \centering

    \begin{subtable}{\textwidth}
        \centering
        \caption{\textbf{Summary statistics}}
        \begin{tabularx}{\textwidth}{l>{\centering\arraybackslash}p{2cm} >{\centering\arraybackslash}p{2cm} 
                                   >{\centering\arraybackslash}p{2cm} >{\centering\arraybackslash}p{2cm} 
                                   >{\centering\arraybackslash}p{2cm} >{\centering\arraybackslash}p{2cm}}
                      & Min  & Q1    & Median & Mean  & Q3    & Max \\ \hline
    Overall           & 5.28 & 16.12 & 21.08  & 21.49 & 25.84 & 47.06 \\ 
    2025--01--01      & 7.28 & 15.35 & 20.09  & 20.54 & 24.80 & 43.66 \\ 
    \Xhline{3\arrayrulewidth}
\end{tabularx}

    \end{subtable}

    \vspace{1em}

    \begin{subtable}{\textwidth}
        \centering
        \caption{\textbf{Number of companies per risk category}}
        \begin{tabularx}{\textwidth}{l>{\centering\arraybackslash}p{2.4cm} >{\centering\arraybackslash}p{2.4cm} 
                                   >{\centering\arraybackslash}p{2.4cm} >{\centering\arraybackslash}p{2.4cm} 
                                   >{\centering\arraybackslash}p{2.4cm}}
                                & Negligible & Low & Medium & High & Severe \\ \hline
            Overall      & 4          & 192 & 222    & 60   & 5 \\ 
            2025--01--01 & 7          & 207 & 218    & 48   & 3 \\ \Xhline{3\arrayrulewidth}
        \end{tabularx}
    \end{subtable}

    \caption*{\small{(a): Minimum, 1st quartile, median, mean, 3rd quartile, and maximum Sustainalytics ESG Risk Ratings over all dates (``overall'') and on the last date (``2025--01--01'').\\
    (b): Number of companies per Sustainalytics ESG Risk Ratings risk category over all dates (``overall'') and on the last date (``2025--01--01'').}}
    \label{tab:summary}
\end{table}

Throughout \cref{sec:risk_measures_for_ESG_assets}, higher ESG ratings correspond to lower risk, whereas in our dataset higher ratings indicate higher ESG risk. We therefore transform and normalize the raw ratings to the interval $[0,1]$. To achieve this, let $S^{\raw}$ denote the raw ESG rating given in the dataset and define
\[
  S^\raw \mapsto S^\norma \defined \frac{50-S^{\raw}}{50}.
\]
This transformation flips the order on the ratings, so that larger values of $S^\norma$ correspond to a better ESG rating. Since the observed ESG ratings range from $5.28$ to $47.06$ (\cref{tab:summary}), it follows that $S^\norma$ is $[0,1]$-valued.
\cref{tab:meanvola} reports descriptive statistics of the stock prices and normalized ESG ratings.

\begin{table}[ht]
\captionsetup{position=above, singlelinecheck=false, labelfont=bf}
\caption{\textbf{Descriptive statistics of stock prices and normalized ESG ratings}}
\centering
  \begin{tabularx}{\textwidth}{l>{\centering\arraybackslash}p{2.7cm} >{\centering\arraybackslash}p{3.8cm} >{\centering\arraybackslash}p{3.0cm} >{\centering\arraybackslash}p{2.8cm}}
                           & Samples  & Mean return   & Volatility \\ \hline
    Stock prices           & 40       & 6.2\%          & 30.6\% \\
    Normalized ESG ratings & 40       & 2.3\%          &  8.1\% \\ \Xhline{3\arrayrulewidth} 
  \end{tabularx}
\caption*{\small{%
  The first column reports the number of monthly observations (October 2021–January 2025), while the second and third columns show the annualized mean and the annualized volatility of monthly log returns (stock prices) and monthly log changes (normalized ESG ratings).}}
\label{tab:meanvola}
\end{table}

Another feature of the dataset is that ESG ratings do not necessarily change every month. Averaged across the 483 companies, ratings change only 9.7 times over the 40-month sample period. This persistence must be taken into account when modeling the distribution of normalized ESG ratings.

As discussed in \cref{rem:additive}, choosing $k=0$ in the utility function $\util$ implies additive independence of the attributes ``financial risk'' and ``ESG risk'', so that dependence between prices and ESG ratings is ignored when computing risk. Put differently, this means that for $k=0$ the joint distribution of prices and ratings does not matter, i.e.\ dependent and independent data is treated identically. \cref{fig:correlation_analysis}, however, shows clear dependence between stock returns and ESG risk ratings, implying that additive independence would discard relevant information in the data and motivating the use of utility functions with $k\neq 0$.

\begin{figure}[hbt!]
\centering
\includegraphics[trim= 0 25px 0 5px, clip, width=\textwidth ]{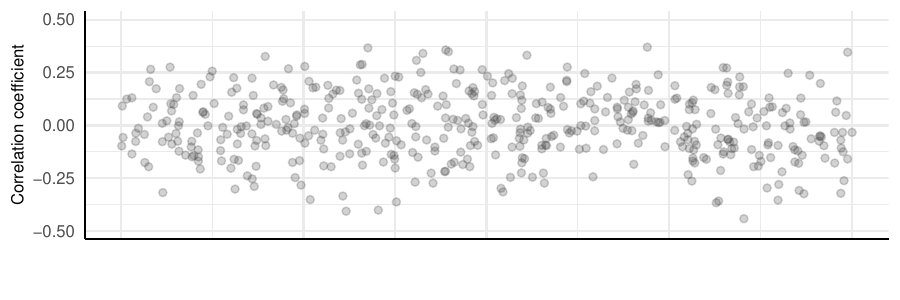} 
\captionsetup{labelfont=bf}
\caption{%
\textbf{Correlation analysis between historical stock price returns and changes in normalized ESG ratings}\\
\small{This figure shows Pearson correlation coefficients between monthly stock price log returns and monthly log changes in normalized ESG ratings. Each dot represents one company.}}
\label{fig:correlation_analysis}
\end{figure}

\subsection{Forecasting of Stock Prices and ESG Ratings}\label{subsec:simulation}

In the following subsections, we compute risk measures for positions $(X,S^\norma)$ over a one-month investment horizon $\Delta\defined 1/12$. To specify the joint distribution, we further rescale normalized ESG ratings $S^\norma$ to ensure comparable codomains of stock prices and ESG ratings by defining
\[
  S^\norma \mapsto S^{\mathrm{rescaled}} \defined \tan\Bigl(\frac{\pi}{2}S^\norma\Bigr),
\]
which maps values to $\R_+$.
Monthly stock log returns and monthly log changes of rescaled ESG ratings are modeled as
\[
  R_X \defined \mu_X \Delta + \sqrt{\Delta} Z_1
  \qquad\text{and}\qquad
  R_S \defined J_S\Bigl(\mu_S \Delta + \sqrt{\Delta} Z_2\Bigr),
\]
where, $\mu_X,\mu_S\in\R$ denote annualized mean returns. The random variable $J_S\in\{0,1\}$ is Bernoulli with
\[
  p\defined\Prob[\{J_S=1\}]\in[0,1]
\]
representing the probability of a change in the ESG rating within one month. The vector $Z=(Z_1,Z_2)^\top$ is drawn from a two dimensional normal distribution with zero mean and covariance matrix
\[
  C \defined \begin{pmatrix}
    \sigma_X^2 & \rho\sigma_X\sigma_S\\
    \rho\sigma_X\sigma_S & \sigma_S^2
  \end{pmatrix},
\]
where $\sigma_X,\sigma_S>0$ are annualized volatilities and $\rho\in[-1,1]$ is the correlation between the two returns conditional on $J_S = 1$, that is, conditional on the ESG rating changing within a month. Although $J_s$ and $Z$ are assumed to be independent, the returns $R_X$ and $R_S$ are clearly correlated due to the dependence of $Z_1$ and $Z_2$.

Let $N\in\R$ denote today's dollar investment in the stock. The financial exposure one month into the future is
\begin{equation}\label{eq:financial_exposure}
  X \defined N\bigl(e^{R_X} - 1\bigr).
\end{equation}
This definition corresponds to log-normal stock prices, as in the Black–Scholes model. Similarly, the rescaled ESG rating after one month is given by
\[
  S^{\mathrm{rescaled}} \defined S^{\mathrm{rescaled}}_0e^{R_S},
\]
where $S^{\mathrm{rescaled}}_0$ denotes today's rescaled ESG rating. Conditional on $J_S = 1$,  $S^{\mathrm{rescaled}}$ is log-normal, and the normalized ESG rating one month into the future is recovered via
\begin{equation}\label{eq:ESG_exposure}
  S^\norma \defined \frac{2}{\pi}\arctan\bigl(S^{\mathrm{rescaled}}\bigr) = \frac{2}{\pi}\arctan\bigl(S^{\mathrm{rescaled}}_0e^{R_S}\bigr).
\end{equation}
The parameters $\mu_X,\mu_S,\sigma_X,\sigma_S,\rho$ and $p$ are calibrated from the dataset discussed in the previous subsection, and we set the investment amount $N \defined \$1$ for simplicity. This fully specifies the distribution of $(X,S^\norma)$ one month into the future.

\subsection{ESG Risk Measures for Single Assets}\label{subsec:ESG_risk_measures_for_single_assets}

We compute ESG risk measures associated with $(X,S^\norma)$ as defined in \cref{eq:financial_exposure} and \cref{eq:ESG_exposure} for each company in the dataset described in \cref{subsec:data_description}.
Specifically, we compare the uncapped entropic ESG risk measure from \cref{ex:entropic_Esg_Risk_measure} with the classical entropic risk measure given by \cref{eq:entropic}.
The entropic ESG risk measure is obtained by choosing
\[
  \util_1(x) \defined \frac{1}{\gamma_1}\Bigl(1-e^{-\gamma_1 x}\Bigr)
  \qquad\text{and}\qquad
  \util_2(s) \defined c\frac{1}{\gamma_2}\Bigl(1-e^{-\gamma_2 (s-s_0))}\Bigr),
  \qquad (x,s)\in\R\times\ESGSpace,
\]
with parameters $\gamma_1,\gamma_2,c>0$ and $s_0\in\ESGSpace = [0,1]$.
In addition, the parameter $k\in\R$ in \cref{eq:utility_uncapped_entropic} is specified.
Parameter values used in the empirical analysis are reported in \cref{tab:parameters}.

\begin{table}[ht]
\centering
\captionsetup{position=above, singlelinecheck=false, labelfont=bf}
\caption{\textbf{Parameter choices for the uncapped entropic ESG risk measure}}
\begin{tabularx}{\textwidth}{%
  >{\centering\arraybackslash}X 
   >{\centering\arraybackslash}X 
   >{\centering\arraybackslash}X
   >{\centering\arraybackslash}X
   >{\centering\arraybackslash}X%
}
    $\gamma_1$  & $\gamma_2$  & $c$ & $k$ & $s_0$\\
    \hline
    $1$         & 0.75        & 0.1 & 1   & 0.5982 \\ \Xhline{3\arrayrulewidth} 
\end{tabularx}
\caption*{\small The baseline level $s_0$ equals the median normalized ESG rating.}
\label{tab:parameters}
\end{table}

Recalling from Example~\ref{ex:calibration} that the multi-attribute utility function $\util$ can be written as
\[
  \util(x,s) = \bigl(1+k\util_2(s)\bigr)\util_1(x) + \util_2(s),\qquad (x,s)\in\R\times\ESGSpace,
\]
with $\ESGSpace = [0,1]$, the parameter choices imply minimum and maximum additive shifts of
\[
  \util_2(0) \approx -0.0755
  \qquad\text{and}\qquad
  \util_2(1) \approx 0.0347,
\]
corresponding to dollar adjustments of the same magnitude for $N = \$1$. The associated minimum and maximum relative scalings of $\util_1$ are
\[
  1+k\util_2(0) \approx 0.9245
  \qquad\text{and}\qquad
  1+k\util_2(1) \approx 1.0347.
\]

Since the expectation in 
\[
  \RiskESG[X,S^\norma]
  \defined \inf\bigl\{m\in\R:\E\bigl[\util(X+m,S^\norma)\bigr]\geq 0\bigr\}
\]
cannot be computed in closed form, we approximate it numerically using $M\defined 10,\!000$ Monte Carlo samples from the distribution of $(X,S^\norma)$.
The infimum in the definition of $\RiskESG[X,S^\norma]$ is then computed using a standard numerical optimizer.

\begin{figure}[ht]
\centering
\includegraphics[trim= 0 5px 0 5px, clip, width=\textwidth]{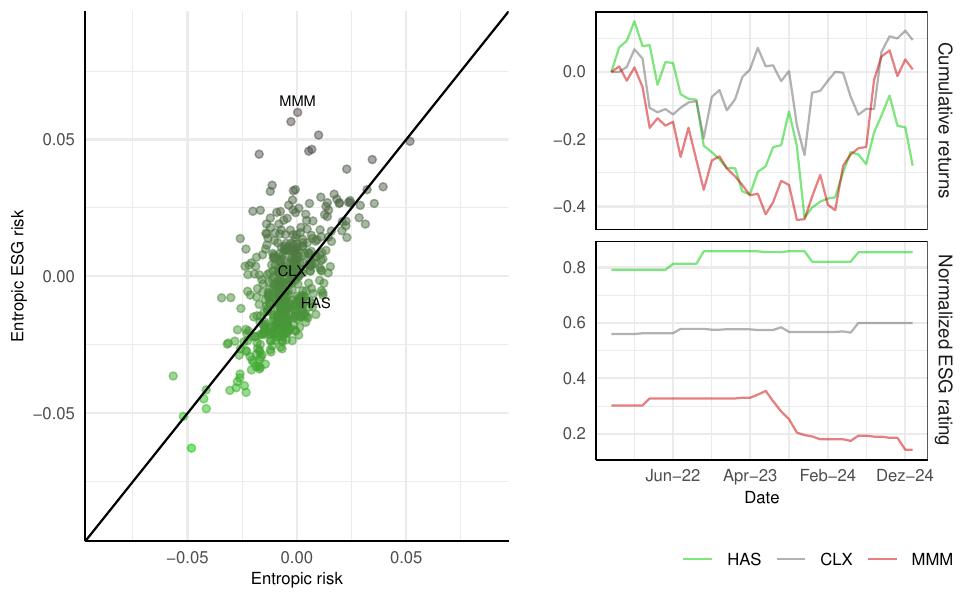}
\captionsetup{labelfont=bf}
\caption{\textbf{Entropic risk and the entropic ESG risk for S\&P 500 companies}\\ \small{Left: Classical entropic risk (x-axis) against entropic ESG risk (y-axis). Right: Cumulative stock price returns and normalized ESG ratings for CLX, HAS, and MMM.}}
\label{fig:scatter_esg_and_classical}
\end{figure}

The results are depicted in \cref{fig:scatter_esg_and_classical}.
The left-hand side compares the classical entropic risk measure $\RiskSF[X]$ with the entropic ESG risk measure $\RiskESG[X,S^\norma]$ across all 483 companies.
Firms on the diagonal line have indifference ESG positions ($\RiskSF[X]=\RiskESG[X,S^\norma]$), while observations above (below) the diagonal exhibit higher (lower) risk once ESG criteria is incorporated.

\begin{table}[ht]
\centering
\captionsetup{position=above, singlelinecheck=false, labelfont=bf}
\caption{\textbf{Entropic risk and entropic ESG risk of representative companies}}
\begin{tabularx}{\textwidth}{%
  l>{\centering\arraybackslash}p{2.3cm}
  >{\centering\arraybackslash}X
  >{\centering\arraybackslash}X
  >{\centering\arraybackslash}X
  >{\centering\arraybackslash}X
  >{\centering\arraybackslash}X%
}
  Company name        & Symbol & Entropic risk     & Entropic ESG risk & Normalized ESG rating & ESG risk category\\
  \hline
  Hasbro              & HAS    & \phantom{-}0.0087 &           -0.0139 & 0.854                 & Negligible       \\
  The Chlorox Company & CLX    &           -0.0022 &           -0.0024 & 0.600                 & Medium           \\
  3M                  & MMM    & \phantom{-}0.0004 & \phantom{-}0.0599 & 0.142                 & Severe           \\
  \Xhline{3\arrayrulewidth}
\end{tabularx}
\caption*{\small
  Entropic risk and entropic ESG risk for representative companies, together with their normalized ESG rating in January 2025 and corresponding ESG risk categories.}
\label{table:representatives}
\end{table}

The right-hand side of \cref{fig:scatter_esg_and_classical} illustrates three representative companies; see also \cref{table:representatives} for the specific values.
Their cumulative returns evolve similarly over October 2021 -- January 2025, resulting in only small differences in classical entropic risk. In contrast, their ESG ratings differ substantially, leading to pronounced differences in entropic ESG risk. MMM exhibits a positive ESG risk premium of $0.0599-0.0004 = 0.0595$ due to its unfavorable ESG rating, while HAS shows a negative ESG premium of $-0.0226$ reflecting its positive rating. For CLX, the ESG risk premium is negligible ($-0.0002$), consistent with its normalized ESG rating of $0.600$ in January 2025 being close to the baseline level $s_0$.

\begin{table}[ht]
\captionsetup{position=above, singlelinecheck=false, labelfont=bf}
\caption{\textbf{Companies with the most significant ESG risk premia}}
\centering

\newcolumntype{B}{X}
\newcolumntype{M}{>{\centering\arraybackslash\hsize=.55\hsize}X}
\newcolumntype{S}{>{\centering\arraybackslash\hsize=.6\hsize}X}

\begin{subtable}{\textwidth}
  \centering
  \caption{\textbf{Largest positive ESG risk premia}}
  \begin{tabularx}{\textwidth}{BMSS}
    Company name    & Entropic risk     & Entropic ESG risk & ESG risk premium\\
    \hline
    Exxon Mobil     &           -0.0173 & 0.0446            & 0.0619  \\
    3M              & \phantom{-}0.0004 & 0.0599            & 0.0595  \\ 
    Apache          &           -0.0027 & 0.0565            & 0.0593  \\
    Chevron         &           -0.0113 & 0.0333            & 0.0446  \\ 
    TransDigm Group &           -0.0202 & 0.0238            & 0.0440  \\ 
    \Xhline{3\arrayrulewidth}
  \end{tabularx}
  \newline
  \vspace*{1em}
\end{subtable}

\begin{subtable}{\textwidth}
  \centering
  \caption{\textbf{Largest negative ESG risk premia}}
  \begin{tabularx}{\textwidth}{BMSS}
    Company name          & Entropic risk     & Entropic ESG risk & ESG risk premium\\
    \hline
    Hasbro                & \phantom{-}0.0087 & -0.0139           & -0.0226         \\
    CDW                   & \phantom{-}0.0008 & -0.0215           & -0.0222         \\
    AvalonBay Communities &           -0.0021 & -0.0235           & -0.0213         \\
    Danaher               & \phantom{-}0.0040 & -0.0168           & -0.0208         \\
    IPG Photonics         & \phantom{-}0.0050 & -0.0152           & -0.0202         \\
    \Xhline{3\arrayrulewidth}
  \end{tabularx}
\end{subtable}
\caption*{\small
Companies with the largest ESG risk premia. Panel (a) shows the largest positive premia, and panel (b) the largest negative premia.}
\label{tab:deviation}
\end{table}

\cref{tab:deviation} summarizes the companies with the most pronounced ESG risk premia. Panel (a) reports the five firms with the largest positive ESG risk premia, reflecting severe risk penalties due to unfavorable ESG ratings, while panel (b) reports the five firms with the largest negative risk premia due to their positive ESG ratings.

To assess the sensitivity of the entropic ESG risk to ESG ratings, we study the effect of deterministic shifts in the normalized ESG rating.  For the representative companies HAS,
CLX, and MMM, we consider
\[
  m \mapsto \RiskESG[X,\min\{\max\{S^\norma+m,0\},1\}],\qquad m\in[-1,1],
\]
which enforces the rating to stay within its range $\ESGSpace=[0,1]$.
The graph of this mapping can be interpreted as the extent to which companies can improve their ESG risk by improving their ESG ratings.

\begin{figure}[hbt!]
\centering
\includegraphics[trim= 0 5px 0 5px, clip, width=\textwidth ]{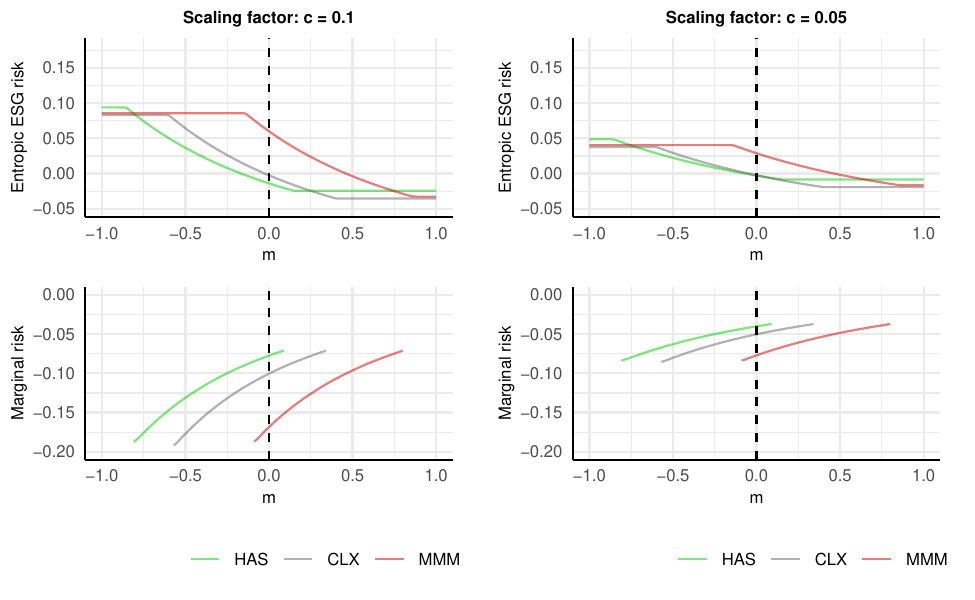}
\captionsetup{labelfont=bf}
\caption{\textbf{Impact of shifts in the normalized ESG rating}\\
\small{%
Entropic ESG risk (top) and marginal entropic ESG risk (bottom) as functions of rating shifts for $c=0.1$ (left) and $c=0.05$ (right).
The plots for the marginal entropic ESG risk are restricted to the range on which the entropic ESG risk is not flat.%
}}
\label{fig:change_in_S}
\end{figure}

\cref{fig:change_in_S} reports the results for two different choices of the scaling parameter $c$ in $\util_2$.
At the boundaries $m\to -1$ and $m\to 1$, the entropic ESG risk becomes flat as
\[
  \RiskESG[X,\min\{\max\{S^\norma+m,0\},1\}] \to \begin{cases}
    \RiskESG[X,0] & \text{as }m\downarrow -1,\\
    \RiskESG[X,1] & \text{as }m\uparrow 1.\\
  \end{cases}
\]
With this in mind, we observe that the representative companies are affected by shifts in the ESG rating in different ways:
Due to its positive ESG rating, HAS has little potential for further risk reduction but strong sensitivity to negative ESG rating shocks, whereas MMM can substantially reduce risk through ESG improvements and displays the largest marginal sensitivity. Overall, firms with unfavorable ESG ratings have greater potential for risk reduction but are also more exposed to adverse rating changes. Lower values of the scaling parameter $c$ reduce this sensitivity by flattening the entropic ESG risk profile.

\begin{figure}[hbt!]
\centering
\includegraphics[trim= 0 5px 0 5px, clip, width=\textwidth ]{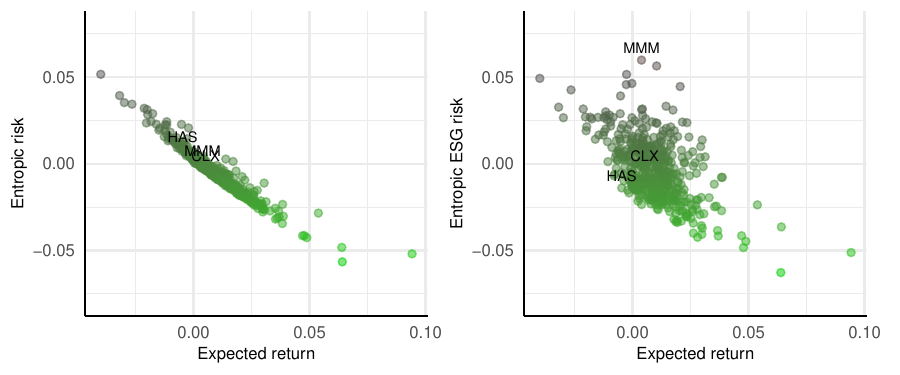}
\captionsetup{labelfont=bf}
\caption{\textbf{%
Entropic risk (left) and entropic ESG risk (right) plotted against the expected return of S\&P 500 companies}\\
\small{%
The plot highlights the risk assessment of the three representative companies HAS, CLX, and MMM relative to their expected returns.
}}
\label{fig:scatter_risk_and_mean}
\end{figure}

We conclude the analysis of individual assets by examining in \cref{fig:scatter_risk_and_mean} the effect of expected returns $\mu_X$ on the risk measures.
While both the classical and ESG risk measure tend to assign less risk to assets with higher returns, the entropic ESG risk exhibits substantially greater dispersion. This indicates that classical entropic risk is largely driven by expected returns, whereas the entropic ESG risk differentiates assets with similar returns based on their ESG performance.

\subsection{Minimum Risk Portfolios}\label{subsec:minimum_risk_portfolios}

We conclude the empirical analysis by studying the impact of ESG risk on minimum risk portfolios. We consider a basket of eleven S\&P 500 stocks (\cref{tab:pf_optimization}), constructed by selecting one stock per industry sector based on market capitalization.

\begin{table}[ht]
\centering
\captionsetup{position=above, singlelinecheck=false, labelfont=bf}
\caption{\textbf{Basket of S\&P 500 stocks considered for portfolio optimization}}
\begin{tabularx}{\textwidth}{%
  l>{\centering\arraybackslash}p{1.7cm}
  >{\centering\arraybackslash}p{2.2cm}
  l>{\centering\arraybackslash}X
  >{\centering\arraybackslash}X%
}
  Company name                  & Symbol & Normalized ESG rating & Sector                \\
  \hline                                                                                   
  Apple                         & AAPL   & 0.664                 & Information technology\\
  American Tower                & AMT    & 0.748                 & Real estate           \\
  Amazon.com                    & AMZN   & 0.478                 & Consumer discretionary\\
  Berkshire Hathaway            & BRKB   & 0.476                 & Financials            \\
  Chevron                       & CVX    & 0.232                 & Energy                \\
  Alphabet                      & GOOGL  & 0.502                 & Communication services\\
  Honeywell International       & HON    & 0.458                 & Industrials           \\
  Johnson \& Johnson            & JNJ    & 0.598                 & Health care           \\
  NextEra Energy                & NEE    & 0.500                 & Utilities             \\
  Procter \& Gamble             & PG     & 0.502                 & Consumer staples      \\
  The Sherwin--Williams Company & SHW    & 0.412                 & Materials             \\
  \Xhline{3\arrayrulewidth}   
\end{tabularx}
\caption*{\small{%
Basket of eleven S\&P 500 stocks used for portfolio optimization, with normalized ESG ratings (January 2025) and sector classifications based on~\cite{sp2025}.%
}}
\label{tab:pf_optimization}
\end{table}

Let $X_i$ denote the one-month financial exposure of a \$1 investment in stock $i$, and let $S_i^\norma$ be its normalized ESG rating one month ahead, defined as in \cref{eq:financial_exposure} and  \cref{eq:ESG_exposure}. Recall that $X_i$ and $S_i^\norma$ are driven by the monthly stock returns $R_X^i$ and the monthly $\log$-changes of rescaled ESG ratings $R_S^i$
\[
  R_X^i \defined \mu_X^i\Delta + \sqrt{\Delta}Z_1^i
  \qquad\text{and}\qquad
  R_S^i \defined J_S^i\Bigl(\mu_S^i\Delta + \sqrt{\Delta}Z_2^i\Bigr),
  \qquad i=1,\dots,11,
\]
where $\mu_X^i,\mu_S^i\in\R$ are annualized means.
In the multi-stock setting, we allow dependence across firms in both the Bernoulli variables $J_S^1,\dots,J_S^{11}$ and the Gaussian shocks $Z_1^1,Z_2^1,\dots,Z_1^{11},Z_2^{11}$.
For the Bernoulli random variables this means that we must specify their pairwise correlations, whereas in the case of the normal random variables we assume that the vector
$(Z_1^1,Z_2^1,\dots,Z_1^{11},Z_2^{11})^\top$ is multivariate normal.
This approach allows us to capture correlations between jump times in the ESG rating, jumps of the ESG ratings, and stock price returns also across different companies.

The investor chooses portfolio weights $w = (w_1,\dots,w_{11})^\top\in\R^{11}$ from the set
\[
  \Xi \defined \Bigl\{w : \sum_{i=1}^{11} w_i = 1\text{ and }0 \leq w_i\leq 0.2\text{ for }i=1,\dots,11\Bigr\},
\]
thus excluding short sales and enforcing diversification.
Given $w\in\Xi$, the portfolio's financial exposure and ESG rating are
\[
  X^w \defined \sum_{i=1}^{11}w_iX_i
  \qquad\text{and}\qquad
  S^w \defined \sum_{i=1}^{11}w_iS^\norma_i.
\]
The objective is to find minimum risk portfolios by solving
\[
  \min_{w\in\Xi} \RiskSF[X^w]
  \qquad\text{and}\qquad
  \min_{w\in\Xi} \RiskESG[X^w,S^w],
\]
where $\RiskSF$ and $\RiskESG$ denote the classical entropic and uncapped entropic ESG risk measures, respectively, with parameters as in \cref{tab:parameters}.

The optimization of portfolios is conducted using a 20-month rolling window. The first minimum risk portfolio is computed in June 2023 using data from October 2021 to June 2023, and the optimization is updated monthly using the preceding 20 months of observations. The final portfolio is obtained in January 2025.

\begin{figure}[hbt!]
\centering
\includegraphics[trim= 0 0px 0 0px, clip, width=\textwidth]{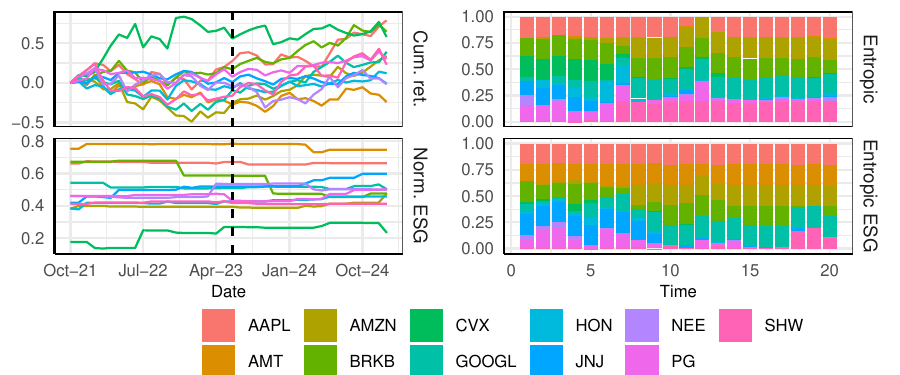}
\captionsetup{labelfont=bf}
\caption{\textbf{%
Cumulative returns and normalized ESG ratings for the basket of S\&P~500 stocks, and optimal portfolio weights of the minimum risk strategies%
}\\
\small{%
Cumulative returns and normalized ESG ratings (left) and optimal portfolio weights (right) for the eleven-stock basket. Portfolio optimization starts in June 2023 (dashed line); weights are shown for classical entropic risk (top) and entropic ESG risk (bottom).
}}
\label{fig:optimal_portfolio_weights_stock_prices_and_ESG_ratings}
\end{figure}

\begin{figure}[hbt!]
\centering

\begin{subfigure}[t]{0.4\textwidth}
    \centering
    \includegraphics[width=\textwidth]{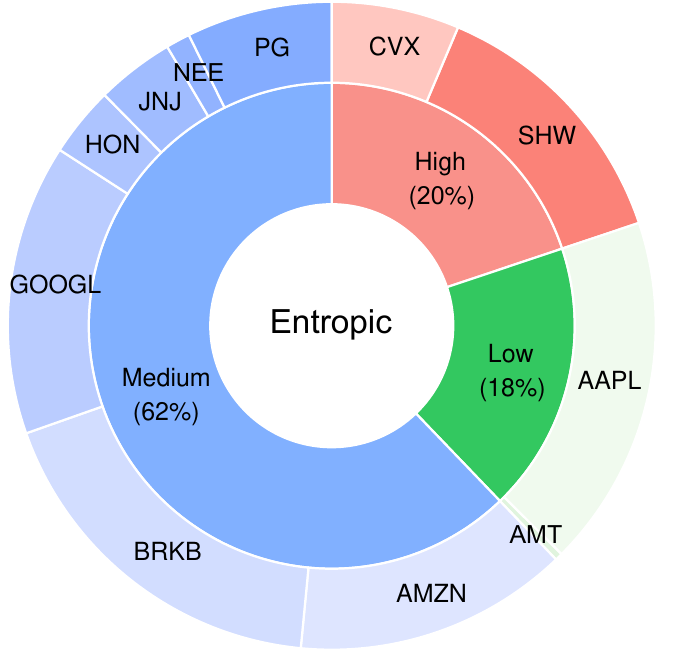}
\end{subfigure}
\hspace{0.1\textwidth}
\begin{subfigure}[t]{0.4\textwidth}
    \centering
    \includegraphics[width=\textwidth]{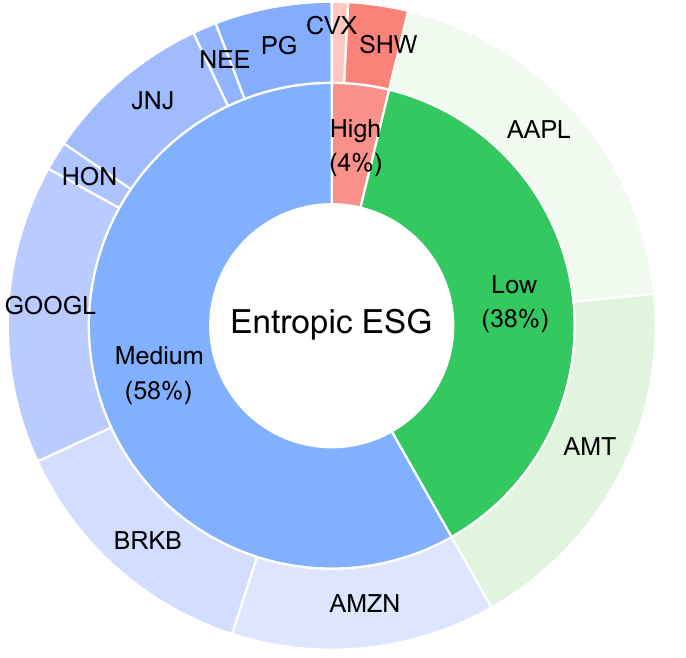}
\end{subfigure}

\captionsetup{labelfont=bf}
\caption{\textbf{Averaged optimal portfolio weights grouped by risk categories}\\
\small{%
Average optimal portfolio weights by ESG risk category under the classical entropic risk measure (left) and the entropic ESG risk measure (right), averaged over the 20-month optimization period.%
}}
\label{fig:weights}
\end{figure}

\begin{table}[ht]
\captionsetup{position=above, singlelinecheck=false, labelfont=bf}
\caption{\textbf{Averaged optimal portfolio weights}}
\centering
\begin{tabularx}{\textwidth}{
    l>{\centering\arraybackslash}p{0.9cm}
    >{\centering\arraybackslash}p{1cm}
    >{\centering\arraybackslash}p{1.1cm}
    >{\centering\arraybackslash}p{1.1cm}
    >{\centering\arraybackslash}p{0.8cm}
    >{\centering\arraybackslash}p{1.3cm}
    >{\centering\arraybackslash}p{0.9cm}
    >{\centering\arraybackslash}p{0.9cm}
    >{\centering\arraybackslash}p{1cm}
    >{\centering\arraybackslash}p{0.9cm}
    >{\centering\arraybackslash}p{1cm}
}
        & AAPL   & AMT    & AMZN  & BRKB    & CVX   & GOOGL  & HON   & JNJ   & NEE   & PG    & SHW   \\
    \hline
    (a) & 17.6\% &  0.4\% & 13.7\% & 18.1\% & 6.4\% & 14.5\% & 3.5\% & 3.9\% & 1.2\% & 7.3\% & 13.5\%\\ 
    (b) & 19.6\% & 18.4\% & 13.2\% & 13.1\% & 0.8\% & 14.9\% & 1.5\% & 8.5\% & 1.2\% & 5.8\% &  3.0\%\\
    \Xhline{3\arrayrulewidth}   
\end{tabularx}
\caption*{\small{%
Average optimal portfolio weights over the 20-month optimization period. Row (a) corresponds to the classical entropic risk measure; row (b) to the entropic ESG risk measure.
}}
\label{tab:weights}
\end{table}

\cref{fig:optimal_portfolio_weights_stock_prices_and_ESG_ratings} shows the resulting optimal portfolio weights, together with cumulative returns and normalized ESG ratings of the constituent stocks.
\cref{tab:weights} reports average portfolio weights over the 20 optimization periods, while Figure~\ref{fig:weights} aggregates weights by ESG risk category.
We observe that both the classical entropic and entropic ESG strategies consistently invest in AAPL, AMZN, BRKB, and GOOGL. This behavior can be explained by the strong financial performance paired with moderate ESG ratings of these companies. 
In contrast, allocations differ significantly for AMT, CVX, JNJ, and SHW. The entropic ESG strategy favors AMT and JNJ, which are among the strongest ESG performers, while avoiding CVX and SHW despite their strong returns due to weak ESG ratings.
Incorporating ESG risk substantially improves the portfolio’s ESG composition. Under the classical entropic strategy, 20\% of the portfolio falls into the ``high'' ESG risk category and 18\% into ``low'', whereas the entropic ESG strategy reduces the high-risk share to 4\% and increases the low-risk share to 38\%.

\begin{figure}[hbt!]
\centering
\includegraphics[trim= 0 0px 0 0px, clip, width=\textwidth ]{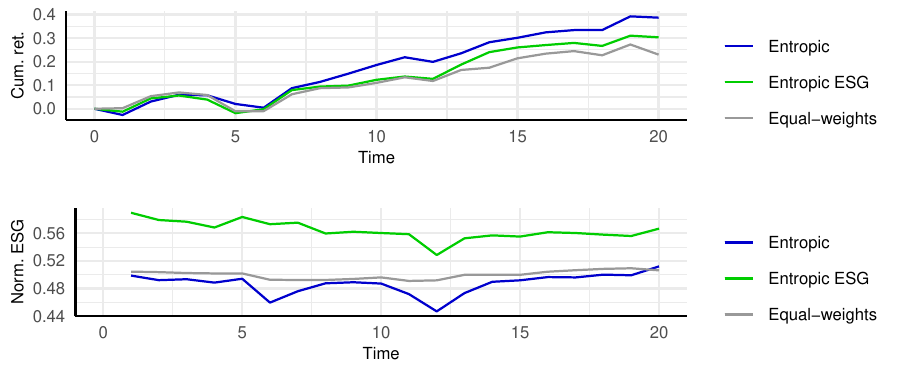}
\captionsetup{labelfont=bf}
\caption{\textbf{%
Cumulative $\boldsymbol\log$-returns and normalized ESG ratings of the minimum risk portfolios and an equal-weights portfolio}\\
\small{%
Evolution of the cumulative log returns (top) and normalized ESG ratings (bottom) of the
three strategies over the 20-month optimization period.%
}
}
\label{fig:financial_performances}
\end{figure}

To compare the minimum risk portfolios’ performance, we present in \cref{fig:financial_performances} the cumulative returns and normalized ESG ratings of the two portfolios and compare them with an equal-weights benchmark. Financial performance is similar across the two optimized strategies, with the classical entropic portfolio slightly outperforming the entropic ESG portfolio.
By the end of the investment period, cumulative log returns are 0.3869 (+47\%) for the classical entropic strategy and 0.3034 (+35\%) for the entropic ESG strategy, compared to 0.2299 (+26\%) for the equal-weights benchmark.

In contrast, ESG performance differs significantly. The entropic ESG strategy achieves a substantially higher normalized ESG rating, exceeding the classical entropic portfolio by 0.0771 and the equal-weights benchmark by 0.0643 on average over the investment period. The classical entropic strategy exhibits the weakest ESG performance, underperforming the equal-weights portfolio by 0.0128. 
We conclude that, within the scope of the minimum risk portfolio setting discussed here, incorporating ESG criteria into portfolio optimization leads to a modest reduction in financial performance but a significant improvement in the overall ESG rating of the portfolio. As such, it constitutes a viable approach for ESG-conscious investors.

\section{Conclusion} \label{sec:conclusion}

This paper introduces a new class of risk measures for ESG assets capable of combining financial risk and risk arising from an asset's ESG rating.
The construction is based on ideas from multi-attribute utility theory, allowing us to build upon axiomatic utility theory to develop ESG risk measures with desirable properties such as translation invariance, monotonicity, and convexity.
In an extensive empirical study, we apply ESG risk measures to an S\&P 500 dataset.
Our results show that ESG risk measures are capable of differentiating assets with similar financial performance based on their ESG assessment.
Moreover, we show how ESG risk measures can be used to improve the ESG performance of an investment portfolio without sacrificing financial performance.

\bibstyle{acm}
\bibliography{bibliography}

\begin{thebibliography}{31}
\providecommand{\natexlab}[1]{#1}
\providecommand{\url}[1]{\texttt{#1}}
\expandafter\ifx\csname urlstyle\endcsname\relax
  \providecommand{\doi}[1]{doi: #1}\else
  \providecommand{\doi}{doi: \begingroup \urlstyle{rm}\Url}\fi

\bibitem[Artzner et~al.(1999)Artzner, Delbaen, Eber, and Heath]{artzner1999}
P.~Artzner, F.~Delbaen, J.-M. Eber, and D.~Heath.
\newblock Coherent measures of risk.
\newblock \emph{Mathematical Finance}, 9\penalty0 (3):\penalty0 203--228, 1999.

\bibitem[Avramov et~al.(2022)Avramov, Cheng, Lioui, and Tarelli]{avramov2022}
D.~Avramov, S.~Cheng, A.~Lioui, and A.~Tarelli.
\newblock Sustainable investing with {ESG} rating uncertainty.
\newblock \emph{Journal of Financial Economics}, 145\penalty0 (2):\penalty0 642--664, 2022.

\bibitem[{Bank of England}(2023)]{boe2023}
{Bank of England}.
\newblock {Bank of England report on climate-related risks and the regulatory capital frameworks}.
\newblock \url{https://www.bankofengland.co.uk/prudential-regulation/publication/2023/report-on-climate-related-risks-and-the-regulatory-capital-frameworks}. Accessed: 2025-04-15., 2023.

\bibitem[Bollen(2007)]{bollen2007}
N.~P. Bollen.
\newblock Mutual fund attributes and investor behavior.
\newblock \emph{Journal of Financial and Quantitative Analysis}, 42\penalty0 (3):\penalty0 683--708, 2007.

\bibitem[Bolton and Kacperczyk(2021)]{bolton2021investors}
P.~Bolton and M.~Kacperczyk.
\newblock Do investors care about carbon risk?
\newblock \emph{Journal of Financial Economics}, 142\penalty0 (2):\penalty0 517--549, 2021.

\bibitem[Capelli et~al.(2021)Capelli, Ielasi, and Russo]{capelli2021forecasting}
P.~Capelli, F.~Ielasi, and A.~Russo.
\newblock Forecasting volatility by integrating financial risk with environmental, social, and governance risk.
\newblock \emph{Corporate Social Responsibility and Environmental Management}, 28\penalty0 (5):\penalty0 1483--1495, 2021.

\bibitem[Capelli et~al.(2023)Capelli, Ielasi, and Russo]{capelli2023}
P.~Capelli, F.~Ielasi, and A.~Russo.
\newblock Integrating {ESG} risks into value-at-risk.
\newblock \emph{Finance Research Letters}, 55:\penalty0 103875, 2023.

\bibitem[Cornell(2021)]{cornell2021esg}
B.~Cornell.
\newblock {ESG} preferences, risk and return.
\newblock \emph{European Financial Management}, 27\penalty0 (1):\penalty0 12--19, 2021.

\bibitem[Dorfleitner and Nguyen(2017)]{dorfleitner2017}
G.~Dorfleitner and M.~Nguyen.
\newblock A new approach for optimizing responsible investments dependently on the initial wealth.
\newblock \emph{Journal of Asset Management}, 18\penalty0 (2):\penalty0 81--98, 2017.

\bibitem[Escobar-Anel(2022)]{escobar2022}
M.~Escobar-Anel.
\newblock Multivariate risk aversion utility, application to {ESG} investments.
\newblock \emph{The North American Journal of Economics and Finance}, 63:\penalty0 101790, 2022.

\bibitem[{European Banking Authority}(2025)]{eba2025}
{European Banking Authority}.
\newblock {Guidelines on the management of ESG risks}.
\newblock \url{https://www.eba.europa.eu/activities/single-rulebook/regulatory-activities/sustainable-finance/guidelines-management-esg-risks}. Accessed: 2025-07-24., 2025.

\bibitem[{European Central Bank}(2020)]{ecb2020}
{European Central Bank}.
\newblock {Guide on climate-related and environmental risks}.
\newblock \url{https://www.bankingsupervision.europa.eu/ecb/pub/pdf/ssm.202011finalguideonclimate-relatedandenvironmentalrisks~58213f6564.en.pdf}. Accessed: 2025-04-15., 2020.

\bibitem[{European Insurance and Occupational Pensions Authority}(2024)]{eiopa2024}
{European Insurance and Occupational Pensions Authority}.
\newblock {Final report on the prudential treatment of sustainability risks for insurers}.
\newblock \url{https://www.eiopa.europa.eu/publications/final-report-prudential-treatment-sustainability-risks-insurers_en}. Accessed: 2025-07-24., 2024.

\bibitem[{Federal Reserve}(2023)]{fed2023}
{Federal Reserve}.
\newblock {Principles for climate-related financial risk management for large financial institutions}.
\newblock \url{https://www.federalreserve.gov/supervisionreg/srletters/SR2309a1.pdf}. Accessed: 2025-07-24., 2023.

\bibitem[F{\"o}llmer and Schied(2002)]{follmer2002}
H.~F{\"o}llmer and A.~Schied.
\newblock Convex measures of risk and trading constraints.
\newblock \emph{Finance and Stochastics}, 6\penalty0 (4):\penalty0 429--447, 2002.

\bibitem[Gallucci et~al.(2022)Gallucci, Santulli, and Lagasio]{gallucci2022}
C.~Gallucci, R.~Santulli, and V.~Lagasio.
\newblock The conceptualization of environmental, social and governance risks in portfolio studies {A} systematic literature review.
\newblock \emph{Socio-Economic Planning Sciences}, page 101382, 2022.

\bibitem[Hamel and Heyde(2010)]{hamel2010duality}
A.~H. Hamel and F.~Heyde.
\newblock Duality for set-valued measures of risk.
\newblock \emph{SIAM Journal on Financial Mathematics}, 1\penalty0 (1):\penalty0 66--95, 2010.

\bibitem[Hartzmark and Sussman(2019)]{hartzmark2019investors}
S.~M. Hartzmark and A.~B. Sussman.
\newblock Do investors value sustainability? {A} natural experiment examining ranking and fund flows.
\newblock \emph{The Journal of Finance}, 74\penalty0 (6):\penalty0 2789--2837, 2019.

\bibitem[Hong and Kacperczyk(2009)]{hong2009price}
H.~Hong and M.~Kacperczyk.
\newblock The price of sin: The effects of social norms on markets.
\newblock \emph{Journal of Financial Economics}, 93\penalty0 (1):\penalty0 15--36, 2009.

\bibitem[{Hong Kong Monetary Authority}(2024)]{hkma2024}
{Hong Kong Monetary Authority}.
\newblock {Enhanced competency framework on green and sustainable finance (professional level)}.
\newblock \url{https://brdr.hkma.gov.hk/eng/doc-ldg/docId/20241231-1-EN}. Accessed: 2025-07-24., 2024.

\bibitem[Hsu et~al.(2023)Hsu, Li, and Tsou]{hsu2023pollution}
P.-H. Hsu, K.~Li, and C.-Y. Tsou.
\newblock The pollution premium.
\newblock \emph{The Journal of Finance}, 78\penalty0 (3):\penalty0 1343--1392, 2023.

\bibitem[Jessen(2012)]{jessen2012}
P.~Jessen.
\newblock Optimal responsible investment.
\newblock \emph{Applied Financial Economics}, 22\penalty0 (21):\penalty0 1827--1840, 2012.

\bibitem[Jouini et~al.(2004)Jouini, Meddeb, and Touzi]{jouini2004vector}
E.~Jouini, M.~Meddeb, and N.~Touzi.
\newblock Vector-valued coherent risk measures.
\newblock \emph{Finance and stochastics}, 8\penalty0 (4):\penalty0 531--552, 2004.

\bibitem[Keeney and Raiffa(1993)]{keeney1993decisions}
R.~L. Keeney and H.~Raiffa.
\newblock \emph{Decisions with Multiple Objectives: Preferences and Value Trade-Offs}.
\newblock Cambridge University Press, 1993.

\bibitem[Kotsantonis et~al.(2016)Kotsantonis, Pinney, and Serafeim]{kotsantonis2016esg}
S.~Kotsantonis, C.~Pinney, and G.~Serafeim.
\newblock {ESG} integration in investment management: Myths and realities.
\newblock \emph{Journal of Applied Corporate Finance}, 28\penalty0 (2):\penalty0 10--16, 2016.

\bibitem[Lins et~al.(2017)Lins, Servaes, and Tamayo]{lins2017social}
K.~V. Lins, H.~Servaes, and A.~Tamayo.
\newblock Social capital, trust, and firm performance: The value of corporate social responsibility during the financial crisis.
\newblock \emph{The Journal of Finance}, 72\penalty0 (4):\penalty0 1785--1824, 2017.

\bibitem[Naffa and Fain(2022)]{naffa2022}
H.~Naffa and M.~Fain.
\newblock A factor approach to the performance of {ESG} leaders and laggards.
\newblock \emph{Finance Research Letters}, 44:\penalty0 102073, 2022.

\bibitem[P{\'a}stor et~al.(2022)P{\'a}stor, Stambaugh, and Taylor]{pastor2022dissecting}
L.~P{\'a}stor, R.~F. Stambaugh, and L.~A. Taylor.
\newblock Dissecting green returns.
\newblock \emph{Journal of Financial Economics}, 146\penalty0 (2):\penalty0 403--424, 2022.

\bibitem[Pedersen et~al.(2021)Pedersen, Fitzgibbons, and Pomorski]{pedersen2021responsible}
L.~H. Pedersen, S.~Fitzgibbons, and L.~Pomorski.
\newblock Responsible investing: {T}he {ESG}-efficient frontier.
\newblock \emph{Journal of Financial Economics}, 142\penalty0 (2):\penalty0 572--597, 2021.

\bibitem[{S\&P Global}(2025)]{sp2025}
{S\&P Global}.
\newblock {U.S.\ Equity.\ S\&P Sectors}.
\newblock \url{https://www.spglobal.com/spdji/en/index-family/equity/us-equity/sp-sectors/#overview}. Accessed: 2025-04-15, 2025.

\bibitem[Torri et~al.(2023)Torri, Giacometti, Dentcheva, Rachev, and Lindquist]{torri2023esg}
G.~Torri, R.~Giacometti, D.~Dentcheva, S.~T. Rachev, and W.~B. Lindquist.
\newblock {ESG}-coherent risk measures for sustainable investing.
\newblock \emph{Preprint, \url{https://arxiv.org/abs/2309.05866}}, 2023.

\end{thebibliography}

\end{document}